\RequirePackage{fix-cm}
\RequirePackage{amsmath}
\documentclass[smallcondensed]{svjour3}     
\smartqed  
\usepackage{graphicx}
\usepackage{latexsym,enumerate,amssymb,amsbsy}
\usepackage{graphics,color}
\usepackage{verbatim}
\usepackage{dblfloatfix}
\usepackage{mathptmx}
\usepackage{algorithm}
\usepackage{algpseudocode}
\usepackage{enumitem}
\usepackage{bm}
\usepackage{tikz}
\usetikzlibrary{fit}
\tikzset{%
  highlight/.style={rectangle,sharp corners,fill=red!15,draw,
    fill opacity=0.3,thin,inner sep=-2pt}
}
\newcommand{\tikzmark}[2]{\tikz[overlay,remember picture,
  baseline=(#1.base)] \node (#1) {#2};}
\newcommand{\Highlight}[1][submatrix]{%
    \tikz[overlay,remember picture]{
    \node[highlight,fit=(left.north west) (right.south east)] (#1) {};}
}
\usepackage{xfrac}
\usepackage{amsfonts}
\usepackage[noadjust]{cite}
\smartqed

\newcommand{\N}{\N}
\newcommand{\Mod}[1]{\ (\text{mod}\ #1)}

\newcommand{\eg}{\emph{e.g.}}
\newcommand{\ie}{\emph{i.e.}}
\newcommand{\Tr}{\mathop{{\rm Tr}}}
\def\F{{\mathbb {F}}}
\def\N{{\mathbb {N}}}
\def\Z{{\mathbb {Z}}}
\def\v{{\mathbf{v}}}
\def\m{{\mathbf{m}}}

\def\u{{\mathbf{u}}}
\def\k{{\mathbf{k}}}
\def\s{{\mathbf{s}}}
\def\0{{\mathbf{0}}}
\def\1{{\mathbf{1}}}
\def\a{{\mathbf{a}}}
\def\b{{\mathbf{b}}}
\def\c{{\mathbf{c}}}

\def\A{{\mathcal{A}}}

\def\alb{{\bm{\alpha}}}

\def\P{{\mathcal{P}}}
\DeclareMathOperator{\lcm}{lcm}
\DeclareMathOperator{\ord}{ord}
\DeclareMathOperator{\rank}{rank}
\begin{document}
\title{The Cycle Structure of LFSR with Arbitrary\\Characteristic Polynomial over Finite Fields}
\author{Zuling Chang \and Martianus Frederic Ezerman \and\\
San Ling \and Huaxiong Wang}
\authorrunning{Chang \and Ezerman \and Ling \and Wang}
\institute{Z. Chang \at School of Mathematics and Statistics, Zhengzhou University, Zhengzhou 450001, China\\
\email{zuling\textunderscore chang@zzu.edu.cn}
\and
M. F. Ezerman \and S. Ling \and H. Wang \at Division of Mathematical Sciences, School of Physical and Mathematical Sciences,\\
Nanyang Technological University, 21 Nanyang Link, Singapore 637371\\
\email{\{fredezerman,lingsan,HXWang\}@ntu.edu.sg}
\and
A preliminary version of this work was presented at SETA 2016.
}
\date{Received: date / Accepted: date}
\maketitle
\begin{abstract}
We determine the cycle structure of linear feedback shift register with arbitrary monic characteristic polynomial over any finite field. For each cycle, a method to find a state and a new way to represent the state are proposed.
\end{abstract}
\section{Introduction}\label{sec:intro}
Let $\F_{q}$, where $q=p^m$ with $m \in \N$, be the finite field of characteristic $p$ and cardinality $q$. A linear feedback shift register (LFSR) with characteristic polynomial $f(x) \in \F_q[x]$ satisfying $f(0) \neq 0$ generates some periodic sequences $\s=(s_0,s_1,\ldots,s_{N-1}) \in \F_{q}^{N}$ with  $N \mid \ord(f(x))$. All polynomials in this work are monic with positive degrees.

Let $\Omega(f(x))$ denote the set containing all sequences with $f(x)$ as their characteristic polynomial. The cycle structure of $\Omega(f(x))$ for arbitrary $f(x)$ reflects the algebraic structure and properties of its LFSR and has many potential applications. Determining the cycle structure is a crucial step in the construction of $q$-ary de Bruijn sequences using the {\it cycle joining method} discussed in, \eg, \cite{Golomb81,Fred82}.

Binary de Bruijn sequences have numerous well-known applications in cryptography, most notably as an ingredient in stream ciphers. Powerful tools developed from de Bruijn graphs whose vertices are length $\ell$ strings of $4$ alphabets accelerated DNA sequencing technologies~\cite{CPT11}. Genome assembly triggered in earnest investigations into possible applications of $q$-ary de Bruijn sequences and related structures, inviting us to go beyond the usual binary set-up.

Two areas where $q$-ary de Bruijn sequences have recently been shown to play important roles are neuroscience and biochemistry. In the former, the sequences are crucial in designing
experiments to test the relationship between ordered stimuli and the neural
responses that they trigger~\cite{AMM11}. In the latter, $4$-ary de Bruijn sequences help build
compact libraries of RNA probes, useful to measure protein-RNA interactions~\cite{OB15}.
These implementations require a \emph{large number} of \emph{efficiently generated sequences}.
Our work lays a good foundation for further investigation in this direction.

Results on the distribution of the periods and the number of sequences having a particular period
can be found in~\cite[Ch.~8, Sect.~5]{LN97}. The exact cycle structure, {\it i.e.}, the list of states belonging to the sequences that generate each cycle in $\Omega(f(x))$, is known only for several limited cases. When $f(x)$ is a primitive polynomial, the cycle structure is treated in~\cite[Ch.~8]{LN97}. For $f(x)=(1+x)^n$, it is given in~\cite{Lempel70}.

Let $\{g_1(x),g_2(x),\ldots,g_k(x)\}$ be a set of $k$ pairwise distinct monic irreducible polynomials in
$\F_q[x]$, $n_i:=\deg(g_i(x))$, and $b_i \in \N$. Let
$f(x):=\prod_{i=1}^{k} g_i^{b_i}(x)$. This paper settles the open problem of determining the cycle structure of $\Omega(f(x))$. For each cycle, a state belonging to the cycle is identified. We then use this knowledge to build a procedure to generate all sequences in $\Omega(f(x))$ by supplying a state belonging to each of these sequences.

\section{Preliminaries}\label{sec:prelims}
A $q$-ary sequence $\s=s_0,s_1,s_2,\ldots$ has entries $s_i \in \F_q$. For $N \in \N$, if $s_{i+N}=s_i$ for all $i \geq 0$, then $\s$ is {\it $N$-periodic} or {\it with period $N$} and we write $\s=(s_0,s_1,\ldots,s_{N-1})$. The period of $\s$ is not unique, so we call the smallest one the {\it least period} of $\s$. For brevity and when it is clear from the context, by {\it period} we mean {\it least period} in the sequel.

Let $\s=s_0,s_1,s_2,\ldots$, $\s'=s'_0,s'_1,s'_2,\ldots$, and $c \in \F_q$. Addition and scalar multiplication are given respectively by $\s+\s'=s_0+s'_0,s_1+s'_1,s_2+s'_2,\ldots$ and $c\s =cs_0,cs_1,cs_2,\ldots$.

It is well-known (see, \eg,~\cite[Ch.~4]{GG05}) that
$\s=(s_0,s_1,\ldots,s_{N-1})$ can be generated by a linear feedback shift register (LFSR).
The entries satisfy the linear recursive relation
\begin{equation}\label{equ:lrr}
s_{k+n}+c_{n-1}s_{k+n-1}+\ldots+c_1s_{k+1}+c_0s_k=0 \mbox{ for } k \geq 0,
\end{equation}
where $c_0,c_1,\ldots,c_{n-1}\in\F_q$, $n\in\N$, and $c_0 \neq 0$.

A {\it (left) shift operator} $L$ acts on $\s=(s_0,s_1,\ldots,s_{N-1})$ by
$L \s=(s_1,s_2,\ldots,s_{N-1},s_0)$. Thus, $L^i \s=(s_i,s_{i+1},\ldots,s_{i-1})$
for $i \in \N$ with the convention that
$L^0 \s=\s$. Using $L$, we write (\ref{equ:lrr}) as
\[
\0 =L^n\s+c_{n-1}L^{n-1}\s+\ldots+c_1L\s+c_0L^0\s =(L^n+c_{n-1}L^{n-1}+\ldots+c_1L+c_0 L^0)~\s.
\]
The equation defines a {\it characteristic polynomial}
\begin{equation}\label{equ:cp}
f(x)=x^n+c_{n-1}x^{n-1}+\cdots+c_1x+c_0
\end{equation}
of $\s$ over $\F_q$. The sequence $\s$ may have many such polynomials. We call
the monic characteristic polynomial with the least degree the {\it minimal polynomial} and
its degree the {\it linear complexity} of $\s$.
The (least) period of $\s$ is equal to the order of its minimal polynomial.

If the minimal polynomial of $\s$ is primitive with degree $n$, then $\s$ is the corresponding $m$-sequence (also known as {\it maximal length sequence}) with period $q^n-1$. This period is the maximal period of all sequences generated by any LFSR with minimal polynomial of degree $n$. We use $\m$ to denote an $m$-sequence.

A sequence $\v$ is said to be a {\it $d$-decimation} sequence of $\s$, denoted by $\v=\s^{(d)}$, if $v_j=s_{d \cdot j}$ for all $j \geq 0$.
A $d$-decimation $\m^{(d)}$ of an $m$-sequence $\m$ is also an $m$-sequence if and only if $\gcd(d,q^n-1)=1$. Hence, the number of distinct $m$-sequences of period $q^n-1$, up to cyclic shifts, is $\frac{\phi(q^n-1)}{n}$ where $\phi(.)$ is the Euler totient function. There is a bijection between the set of all such sequences and the set of primitive polynomials of degree $n$ in $\F_{q}[x]$. More properties of sequences in relation to their characteristic and minimal polynomials can be found in~\cite[Ch.~4]{GG05} and \cite[Ch.~8]{LN97}.

For $\s=(s_0,s_1,\ldots,s_{N-1})$ with characteristic polynomial $f(x)$ defined in (\ref{equ:cp}),
the vector $\s_i=(s_i,s_{i+1},\ldots,s_{i+n-1})\in\F_q^n$ is the $i$-th $n$-stage state of $\s$. In particular, we call $\s_0$ the {\it initial state}. Note that the states
of the sequence depends on the specified characteristic polynomial.

A {\it state operator} $T$ turns the state $\s_i$ into $\s_{i+1}$, {\it i.e.}, $\s_{i+1}=T \s_i$. If $\s_i$ is a state of $\s$ and $e$ is the period of $\s$, then the $e$ \emph{distinct} states of $\s$ are
$\s_i, T \s_i = \s_{i+1},\ldots, T^{e-1} \s_i = \s_{i+e-1}$.

To $f(x)$ in (\ref{equ:cp}), one associates a matrix
\begin{equation}\label{comx}
A:=\begin{pmatrix}
    0 & 0  & \cdots & 0 & -c_0 \\
    1 & 0  & \cdots & 0 & -c_1 \\
    0 & 1  & \cdots & 0 & -c_2 \\
    \vdots & \vdots & \ddots & \vdots & \vdots\\
    0 & 0  & \cdots & 1 & -c_{n-1}
   \end{pmatrix}.
\end{equation}

Using $\s_0$ as the input state, the state vectors of the resulting sequence are $\s_{j}=\s_0 A^{j}$ for $j \in \{0,1,2,\ldots\}$.
We know from~\cite[Thm. 8.13]{LN97} that $e$ divides the order of $A$ in $GL(n,\F_q)$.
Furthermore, $\s_{i+1}=\s_iA=T\s_i$.

For an $N$-periodic sequence $\u \in\Omega(f(x))$, define the set
$[\u]:=\left\{\u,L \u,L^2 \u,\ldots,L^{N-1} \u \right\}$ to be a {\it shift equivalent class} or a {\it cycle} in $\Omega(f(x))$. Sequences in a cycle are said to be {\it shift equivalent}. Thus, $\Omega(f)$ can be partitioned into cycles.

If $\Omega(f(x))$ consists of exactly $j$ cycles $[\u_1],[\u_2],\ldots, [\u_j]$ for some $j \in \N$, then its {\it cycle structure} is presented as $\Omega(f(x))=[\u_1] \cup [\u_2] \cup \ldots \cup [\u_j]$. To study the cycle structure, determining a state belonging to
each sequence $\u_i$ is crucial. Starting with the state, one easily generates the whole sequence and, hence, the corresponding cycle in $\Omega(f(x))$.

Given our setup, from \cite[Sect. 4.4]{GG05} we can derive the following properties.
\begin{enumerate}
\item $\Omega(f(x))$ can be decomposed into a direct sum of $\Omega(g_i^{b_i}(x))$ with $1\leq i\leq k$, \ie,
\[
\Omega(f(x))=\Omega(g_1^{b_1}(x))\oplus\Omega(g_2^{b_2}(x))\oplus\cdots\oplus\Omega(g_k^{b_k}(x)).
\]
\item Every sequence $\s\in\Omega(f(x))$ can be written as
\[
\s=\s_1+\s_2+\cdots+\s_k \mbox{ with } \s_i \in \Omega(g_i^{b_i}(x)) \mbox{ for } 1 \leq i \leq k.
\]
\item The period of $\s$ is $\lcm(e_1,e_2,\ldots,e_k)$ where $e_i$ is the period of $\s_i$.
\end{enumerate}

\begin{remark}
The polynomial $f(x)$ in~\cite[Thm. 4.10]{GG05} is a product of distinct irreducible polynomials, \ie, $b_i=1$ for all $i$. The same inductive proof works in our setup here.
\end{remark}

\section{The Cycle Structure of $\Omega(g(x))$ when $g(x)$ is Irreducible}\label{sec:c1}
This section discusses two ways to determine a state belonging to a sequence with irreducible minimal polynomial
\begin{equation}\label{eq:g(x)}
g(x)=x^n+g_{n-1}x^{n-1}+\cdots+g_1x+g_0 \in \F_{q}[x].
\end{equation}
In the first subsection the state is found by using cyclotomic classes while the second subsection makes use of decimation sequences.

\subsection{Using Cyclotomic Classes}\label{subsec:cyclo}
The relation between sequences and cyclotomic classes can be used to find a state for each sequence with minimal polynomial $g(x)$.
Let $\beta \in \F_{q^n}$ be a root of $g(x)$ with order $e$ and let $t=\frac{q^n-1}{e}$.
There must exist a primitive polynomial
\[
q(x)=x^n+q_{n-1}x^{n-1}+ \ldots + q_0 \in \F_{q}[x]
\]
with a root $\alpha$ satisfying $\beta=\alpha^t$. While we need only one $q(x)$ to be associated with a given $g(x)$, in general the choice is not unique.

\begin{lemma}\label{lem:num-pri}
Let $g(x)$ be a non-primitive irreducible polynomial of degree $n$ with a root $\beta$ and order $e$. There are $\frac{\phi(q^n-1)}{\phi(e)}$ primitive polynomials that can be associated with $g(x)$. Such a polynomial
$q(x)$ has degree $n$ with a root $\alpha$ that satisfies $\beta=\alpha^t$.
\end{lemma}

\begin{proof}
Each of the roots of $q(x)$, say $\chi$, satisfies $g(\chi^t)=0$. It is then clear that a chosen $q(x)$ can be
an associated primitive polynomial of only one irreducible polynomial of degree $n$ and order $e$. Two distinct non-primitive irreducible
polynomials, both with degree $n$ and order $e$, have the same number of associated primitive polynomials.
There are $\phi(q^n-1)/n$ primitive polynomials with degree $n$ and there are $\phi(e)/n$ non-primitive irreducible polynomials with
degree $n$ and order $e$. Taking their ratio completes the proof.\qed
\end{proof}

Following \cite{Storer67}, the {\it cyclotomic classes} $C_i \subseteq \F_{q^n}$, for $0 \leq i <t$, are defined by
\begin{equation}\label{eq:cyclas}
C_i=\{\alpha^{i+ k \cdot t}~|~0\leq k <e\}=\{\alpha^i\beta^k~|~0\leq k<e\}=\alpha^i C_0.
\end{equation}
The classes partition $\F_{q^n}$ into disjoint union
\[
\F_{q^n}=\{0\}\cup C_0\cup\cdots\cup C_{t-1}.
\]

Using $\{1,\beta,\ldots,\beta^{n-1}\}$ as a basis for $\F_{q^n}$ as an $\F_q$-vector space,
for $0\leq j < q^n-1$, one can uniquely express $\alpha^j$ as
$\alpha^j=\sum_{i=0}^{n-1}a_{j,i}\beta^i \text{ with } a_{j,i}\in \F_q$.
Hence,
\begin{equation}\label{eq:alphatobeta}
\alpha^{j+t}=\alpha^j\beta=\beta\sum_{i=0}^{n-1}a_{j,i}\beta^i=\sum_{i=0}^{n-1}a_{j,i}\beta^{i+1}
=\sum_{i=0}^{n-2}a_{j,i}\beta^{i+1} - \sum_{i=0}^{n-1}a_{j,n-1}g_i\beta^i.
\end{equation}

Define the mapping $\varphi:\F_{q^n}\rightarrow \F_q^{n}$ by
$\varphi(0)=\0$ and $\varphi(\alpha^{j})=(a_{j,0},a_{j+t,0},\ldots,a_{j+(n-1)t,0})$ with the subscripts reduced modulo $(q^n-1)$. By the recursive relation in (\ref{eq:alphatobeta}), $\varphi$ is a one-to-one mapping. Let
\begin{equation}\label{eq:corres}
\u_i := \left(a_{i,0},a_{i+t,0},\ldots,a_{i+(e-1)t,0}\right).
\end{equation}

The discussion in~\cite[Thm. 3]{HH96} treats the case of $q=2$. It generalizes naturally to any $q$. Following the same
line of arguments, we deduce that, under $\varphi$, the cyclotomic class $C_i$ corresponds to the cycle $[\u_i]$.
In other words, $\u_i$ and the sequence of states of $\u_i$, namely $((\u_i)_0,(\u_i)_1,\ldots,(\u_i)_{e-1})$, where
$(\u_i)_j=\left(a_{i+jt,0},a_{i+(j+1)t,0},\ldots,a_{i+(j+n-1)t,0}\right)=\varphi(\alpha^i\beta^j)$
for $0\leq j<e$, are shift equivalent. Hence, $\u_i \longleftrightarrow C_i$. Thus,
\begin{equation}\label{equ:g}
\Omega(g(x))=[\0]\cup[\u_0]\cup[\u_1]\cup\ldots\cup[\u_{t-1}].
\end{equation}

Since $\{1,\alpha,\ldots,\alpha^{n-1}\}$ and $\{1,\beta,\ldots,\beta^{n-1}\}$ are bases for $\F_{q^n}$ as an $\F_q$-vector space, there is an invertible matrix $M$ such that $\left(1,\beta,\ldots,\beta^{n-1}\right)=\left(1,\alpha,\ldots,\alpha^{n-1}\right)M$.

Build $M$ by using the representations of $1,\beta,\ldots,\beta^{n-1}$ in the basis $\{1,\alpha,\ldots,\alpha^{n-1}\}$ and the fact that $q(\alpha)=0$. Conversely, with $\{1,\beta,\ldots,\beta^{n-1}\}$ as the basis, $(1,\alpha,\ldots,\alpha^{n-1})=(1,\beta,\ldots,\beta^{n-1})M^{-1}$, from which we derive the representation of $\alpha^i$ for $0\leq i<q^n-1$.

For $0\leq i<t$, using the representation of $\alpha^i$ in the basis $\{1,\beta,\ldots,\beta^{n-1}\}$ and by (\ref{eq:alphatobeta}),
we derive $\varphi(\alpha^i)$, which is a state of sequence $\u_i$.

\subsection{Via Decimation}\label{subsec:decimate}

An alternative method to find a state for each nonzero cycle in $\Omega(g(x))$ involves decimating $m$-sequences. The {\it trace function} from $\F_{q^n}$ to $\F_q$ is given by $\Tr(x)=x+x^q+\ldots+x^{q^{n-1}}$. Recall, \eg, from~\cite[Sect.~4.6]{GG05} that $\m=(m_0,m_1,\ldots,m_{q^n-2})$ whose characteristic polynomial $q(x)$ is primitive of degree $n$ with a root $\alpha$ can be described using
\[
m_i=\Tr(\gamma \alpha^i): 0 \neq \gamma \in \F_{q^n} \mbox{ for } i=0,1,2,\ldots,q^n-2.
\]
Without loss of generality, one can let $\gamma=1$, \ie, $m_i=\Tr(\alpha^i)$.

From $\m$, construct the $t$ distinct $t$-decimation sequences, each of period $e$:
\[
\u_0=\m^{(t)},\u_1=(L\m)^{(t)},\ldots,\u_{t-1}=(L^{t-1}\m)^{(t)}.
\]
Observe that the entries in the resulting sequences satisfy
\[
(\u_k)_j= \Tr(\alpha^{k+ t\cdot j}) \mbox{ for } 0 \leq k <t \mbox{ and } 0 \leq j <e.
\]
Since $\beta=\alpha^t$, each cycle $[\u_i]$ is a cycle in $\Omega(g(x))$.

Starting from an arbitrary $n \cdot t$ consecutive elements of $\m$, one can derive $t$ distinct
$n$-stage states by the above $t$-decimating process. It is then straightforward to verify that each of the derived states corresponds to one nonzero cycle. Since the ordering of the cycles in $\Omega(g(x))$ has no significant bearing in our investigation, we can label these states arbitrarily. 

Given an irreducible polynomial $g(x)$, finding an associated primitive polynomial $q(x)$ can become computationally expensive for large values of $n$ and $t$. Decimating $m$-sequences is a useful tool to accomplish the task while keeping the time complexity low.

Let $\lambda:=\frac{\phi(q^n-1)}{n}$. Given $g(x)$, for $1 \leq j \leq \lambda$ let $p_j(x)$ be a primitive polynomial of degree $n$ and let $\m_j$ be its corresponding $m$-sequence.
The set $\mathcal{M}_n$ of all shift inequivalent $m$-sequences with period $q^n-1$ is given by
\[
\mathcal{M}_n=\{\m_1,\m_2,\ldots,\m_{\lambda}\}.
\]
One can also build $\mathcal{M}$ from an arbitrary $m$-sequence $\m$ by simply collecting all 
the $d_i$-decimations $\{\m^{(d_i)}\}$ with $d_i$ satisfying $\gcd(d_i,q^n-1)=1$ and, 
for $i \neq j$, $d_i$ and $d_j$ do not belong to the same conjugate coset.

We now perform a systematic check to determine which among the $p_j(x)$s can be used as the associated primitive polynomials of $g(x)$. Derive $\m_{j}^{(t)}$ with period $e$ and check if it shares a common string of $2n$ consecutive elements with a sequence generated by an LFSR with characteristic polynomial $g(x)$. If yes, then $p_j(x)$ can be associated to $g(x)$. The states that we want can then be determined using the method already described above. Performing the routine for all $\m_j$s guarantees a match between $g(x)$ and some $p_j(x)$. The required steps can be performed without using costly operations over $\F_q$. As $q$ or $n$ grows, our saving becomes more prominent since operations on 
polynomials over large finite fields are prohibitively expensive.

To end this section, the following example follows both approaches.

\begin{example}\label{example1}
Let $g(x)=x^3+2x+2\in\F_3[x]$. It is irreducible with a root $\beta$ of order $13$. Hence, $(e,t)=(13,2)$ and $\Omega(g(x))=[\0]\cup[\u_0]\cup[\u_1]$.
Note that $q(x)=x^3+x^2+2x+1\in\F_3[x]$ is a primitive polynomial whose root $\alpha$ satisfies $\beta=\alpha^2$. We can then express
\[
(1,\alpha,\alpha^2)=(1,\beta,\beta^2)
\begin{pmatrix}
    1 & 2 & 0  \\
    0 & 1 & 1  \\
    0 & 1 & 0
\end{pmatrix}.
\]
Using $\beta^3=1+\beta$ and writing in the form of (\ref{eq:alphatobeta}),
$\alpha =2+\beta+\beta^2, \alpha^3 =1+\beta^2$, and $\alpha^5=1+2\beta$.
This implies that $\a_1=\varphi(\alpha)=(2,1,1)$ is a state of $\u_1$ and
$\a_0=(1,0,0)$ is a state of $\u_0$. Thus, $\u_0=(1,0,0,1,0,1,1,1,2,2,0,1,2)$ 
and  $\u_1=(2,1,1,0,2,1,2,0,0,2,0,2,2)$.

We can also get the respective states of $\u_0$ and $\u_1$ via decimation. The
$m$-sequence with characteristic polynomial $q(x)$ is
\[
\m=(1,1,1, 2,1,0, 2,0,2, 2,0,0, 1,2,2, 2,1,2, 0,1,0, 1,1,0, 0,2).
\]
The first $6$ consecutive elements of $\m$ can be taken to form $\v=(1,1,1,2,1,0)$.
Hence, $\v^{(2)}=(1,1,1)$ and $(L \v)^{(2)}=(1,2,0)$ are the respective $3$-stage
states of $\u_0$ and $\u_1$.

In $\F_3[x]$, there are $4$ primitive polynomials of degree $3$: $p_1(x)=x^3 + 2x + 1
\mbox{, } p_2(x)=x^3 + 2x^2 + x + 1 \mbox{, } p_3(x)=x^3 + x^2 + 2x + 1
\mbox{, and } p_4(x)=x^3 + 2x^2 + 1$. There are also $4$ non-primitive
irreducible polynomials of degree $3$, each having $e=13$ and $t=2$:
$g_1(x)=x^3 + 2x^2 + 2x + 2 \mbox{, }
g_2(x)=x^3 + 2x + 2 \mbox{, } g_3(x)=x^3 + x^2 + 2 \mbox{, and }
g_4(x)= x^3 + x^2 + x + 2$. By Lemma~\ref{lem:num-pri}, each $g_i(x)$ is associated
with exactly one $p_j(x)$ for $1\leq i,j \leq 4$. We now use decimation to find the
association and, subsequently, determine the respective states of the nonzero cycles in $\Omega(g_i(x))$ for all $i$.

The $m$-sequences are
\begin{align*}
\m_1&=(1, 1, 1, 0, 0, 2, 0, 2, 1, 2, 2, 1, 0, 2, 2, 2, 0, 0, 1, 0, 1, 2, 1, 1, 2, 0),\\
\m_2&=(1, 1, 1, 2, 0, 0, 1, 1, 0, 1, 0, 2, 1, 2, 2, 2, 1, 0, 0, 2, 2, 0, 2, 0, 1, 2),\\
\m_3&=(1, 1, 1, 2, 1, 0, 2, 0, 2, 2, 0, 0, 1, 2, 2, 2, 1, 2, 0, 1, 0, 1, 1, 0, 0, 2) \mbox{, and}\\
\m_4&=( 1, 1, 1, 0, 2, 1, 1, 2, 1, 0, 1, 0, 0, 2, 2, 2, 0, 1, 2, 2, 1, 2, 0, 2, 0, 0 ).
\end{align*}
Their respective $2$-decimation sequences are
\begin{align*}
\m_1^{(2)}&=(1, 1, 0, 0, 1, 2, 0, 2, 0, 1, 1, 1, 2)\mbox{, }
\m_2^{(2)}=(1, 1, 0, 1, 0, 0, 1, 2, 1, 0, 2, 2, 1),\\
\m_3^{(2)}&=(1, 1, 1, 2, 2, 0, 1, 2, 1, 0, 0, 1, 0) \mbox{, }
\m_4^{(2)}=( 1, 1, 2, 1, 1, 1, 0, 2, 0, 2, 1, 0, 0).
\end{align*}
We feed the first three entries of $\m_j^{(2)}$ to each LFSR whose characteristic polynomial is $g_i(x)$ to establish the association and the results given in Table~\ref{table:ex}.

\begin{table}[h!]
\caption{States and Nonzero Cycles in $\Omega(g_i(x))$}
\label{table:ex}
\centering
\renewcommand{\arraystretch}{1.1}
\begin{tabular}{cc|c|c}
\hline
$g(x)$ & $q(x)$ & States & Nonzero Cycles \\
\hline
$g_1(x)$ & $p_4(x)$ & $(1,1,2),(1,0,1)$ & $[(1, 1, 2, 1, 1, 1, 0, 2, 0, 2, 1, 0, 0)],
[(1, 0, 1, 2, 0, 0, 2, 2, 1, 2, 2, 2, 0)]$ \\

$g_2(x)$ & $p_3(x)$ & $(1,1,1),(1,2,0)$ & $[(1, 1, 1, 2, 2, 0, 1, 2, 1, 0, 0, 1, 0)],
[(1, 2, 0, 0, 2, 0, 2, 2, 2, 1, 1, 0, 2)]$ \\

$g_3(x)$ & $p_2(x)$ & $(1,1,0),(1,2,0)$ & $[(1, 1, 0, 1, 0, 0, 1, 2, 1, 0, 2, 2, 1)],
[(1, 2, 0, 1, 1, 2, 2, 2, 0, 2, 0, 0, 2)]$ \\

$g_4(x)$ & $p_1(x)$ & $(1,1,0),(1,0,2)$ & $[(1, 1, 0, 0, 1, 2, 0, 2, 0, 1, 1, 1, 2)],
[(1, 0, 2, 2, 2, 1, 2, 2, 0, 0, 2, 1, 0)]$ \\

\hline
\end{tabular}
\end{table}
\qed
\end{example}

\section{The Cycle Structure of $\Omega\left(g^b(x)\right)$ with $b>1$}\label{sec:c2}
Our next task is to find one state of each cycle in $\Omega\left(g^b(x)\right)$ with
$g(x), \beta, e$, and $t$ as defined in Sect.~\ref{sec:c1} and $1 < b \in \N$. Since (\ref{equ:g}) already gives the cycle structure of $\Omega(g(x))$, we consider the states of the cycles in $\Omega\left(g^b(x)\right)\setminus \Omega(g(x))$. The following result follows directly
from~\cite[Thm. 3.8 and Thm. 8.63]{LN97}.
\begin{lemma}\label{lem:cycle}
For $1\leq r\leq b$, sequences in $\Omega\left(g^b(x)\right)$ with minimal polynomial $g^r(x)$
can be divided into
\[
\frac{q^{nr}-q^{n(r-1)}}{e \cdot p^c}=t\cdot\frac{q^{n(r-1)}}{p^c}
\]
distinct cycles, each with period $e \cdot p^c$ where $c$ is the least integer 
satisfying $p^c \geq r$.
\end{lemma}

We adapt the notion of {\it ${\rm D}$-morphism}, first introduced by Lempel in~\cite{Lempel70}, to generate cycles with minimal polynomial $g^{r+1}(x)$ from cycles with minimal polynomial $g^{r}(x)$. A state of a cycle with minimal polynomial $g^{r+1}(x)$ can then be derived from a state belonging to a cycle with minimal polynomial $g^{r}(x)$. We apply the procedure recursively to cover all cycles. The mapping ${\rm D}: \F_q^{(r+1)n} \longrightarrow \F_q^{rn}$ sends $\a=(a_0,a_1,\ldots,a_{(r+1)n-1}) \mapsto \b=(b_0,b_1,\ldots,b_{rn-1})$
\begin{equation}\label{equ:defD}
\iff
\begin{cases}
b_0         &= g_0 a_0     +g_1 a_1      +\ldots   + g_{n-1} a_{n-1}  +a_n,    \\
b_1         &=g_0 a_1      +g_1 a_2      +\ldots   + g_{n-1} a_{n}   +a_{n+1},\\
~\vdots      &~\vdots                                          \\
b_{rn-1}    &=g_0 a_{rn-1} +g_1 a_{rn}   +\ldots   + g_{n-1} a_{(r+1)n-2} +a_{(r+1)n-1}.
\end{cases}
\end{equation}

Note that ${\rm D}$ is onto. Given $(a_0,a_1,\ldots,a_{n-1}) \in \F_q^n$ and $\b \in \F_q^{rn}$,
one can uniquely determine $a_{n},a_{n+1},\ldots,a_{(r+1)n-1}$. Hence, for any $\b$, there
are $q^n$ distinct $\a$'s satisfying $\b= {\rm D}(\a)$.

Let $\a=\left(a_0,a_1,\ldots,a_{(r+1)n-1}\right) \in \F_q^{(r+1)n}$ be an $n(r+1)$-stage state
of a sequence with characteristic polynomial $g(x)$. Then 
$g_0 a_i + g_1 a_{i+1} + \ldots + a_{i+n}=0$ implies ${\rm D}(\a)=\0$ for $0 \leq i \leq nr-1$.

\begin{lemma}\label{lemD1}
The mapping ${\rm D}$ is a surjective homomorphism whose kernel, denoted by $\ker({\rm D})$, is the set
$\{\k_0,\k_1,\ldots,\k_{q^n-1}\}$ of $q^n$ distinct $n(r+1)$-stage states of the sequences in $\Omega(g(x))$.
\end{lemma}

Since $L({\rm D}(\a))={\rm D}(L(\a))$ for any $\a \in \F_q^{(r+1)n}$, ${\rm D}$ preserves shift relation. Whenever $\a_0,\a_1,\ldots$ are consecutive states of a sequence, we know that ${\rm D}(\a_0), {\rm D}(\a_1),\ldots$ are also consecutive states of some sequence. Hence, for a sequence $\s=(s_0,s_1,\ldots,s_{N-1})$ with period $N$, sequence $\u:={\rm D}(\s)$ has period $N' \mid N$ and entries
\begin{equation}\label{equ:defu}
u_i=g_0 s_i+g_1 s_{i+1}+\ldots+g_{n-1} s_{i+n-1}+s_{i+n}.
\end{equation}

\begin{lemma}\label{lemD2}
If $\s \in \Omega(g^{r+1}(x))$, then $\u:={\rm D}(\s) \in \Omega(g^{r}(x))$. If $g^c(x)$ is the minimal polynomial of $\s$, for some $1 \leq c \leq r+1$, then $g^{c-1}(x)$ is the minimal polynomial of $\u$.
\end{lemma}

\begin{proof}
For a given sequence $\s$ in $\Omega\left(g^{r+1}(x)\right)$,
sequence $\u$ in (\ref{equ:defu}) can be written as
\[
\u=g_0 \s+g_1 L \s+\ldots+g_{n-1} L^{n-1} \s +L^n \s= g(L) \s.
\]
Since $g^{r+1}(L) \s=\0$, we have $g^{r}(L)\u=g^{r}(L)~g(L) \s=g^{r+1}(L)\s=\0$.
Thus, $\u \in \Omega\left(g^{r}(x)\right)$.

The second assertion can be similarly argued using the fact that $g(x)$ is irreducible.\qed
\end{proof}

\begin{corollary}
Let $\s$ be a sequence with minimal polynomial $h(x)$. Then the minimal polynomial of $\u={\rm D}(\s)$
is
\begin{equation*}
\begin{cases}
\frac{h(x)}{g(x)} & \mbox{if } g(x) \mbox{ divides } h(x) \text{ and}\\
h(x) & \mbox{if } g(x) \mbox{ does not divide } h(x).
\end{cases}
\end{equation*}
\end{corollary}

Let $\b \in \F_q^{nr}$ be a state belonging to sequence $\u$ with minimal polynomial $g^{r}(x)$. By (\ref{equ:defD}), the $q^n$ distinct vectors $\a$ satisfying $\b= {\rm D}(\a)$ are the respective states of some sequences with minimal polynomial $g^{r+1}(x)$. Given one such $\a$, these states can be written as the sum of $\a$ and the corresponding states in $\ker({\rm D})$, {\it i.e.}, as $\a+\mathbf{k}_i$ with $0 \leq i < q^n$.

Let $\s$ be a sequence with minimal polynomial $g^{r+1}(x)$ containing $\a$ as a state. Now, consider all $q^n$ sequences: $\s$ plus each of the $q^n$ distinct sequences in $\Omega(g(x))$. We next decide whether these sequences correspond to distinct cycles with minimal polynomial $g^{r+1}(x)$. The consideration is based on whether $r$ is a power of $p$.

\begin{lemma}
Suppose $r$ is {\bf not} a power of $p$. Let $\s$ be a sequence with minimal polynomial $g^{r+1}(x)$. Let $\s_1$ and $\s_2$ be two distinct sequences in $\Omega(g(x))$. Then $\s+\s_1$ and $\s+\s_2$ are shift inequivalent.
\end{lemma}

\begin{proof}\label{lemr1}
Lemma~\ref{lem:cycle} says that the period of $\s$ is $e \cdot p^c$ where $c$ is the least integer with $p^c \geq r+1$. Assume that $\s+\s_1$ and $\s+\s_2$ are shift equivalent. Then there exists an integer $0 \leq \ell <e \cdot p^c$ such that $\s+\s_1=L^{\ell}(\s+\s_2)=L^{\ell}\s + L^{\ell}\s_2$.
Hence, $\left(L^{\ell}-I \right)\s = \s_1 - L^{\ell} \s_2 \in \Omega\left(g(x)\right)$.
Thus, $g(x)$ must be a characteristic polynomial of $(L^{\ell}-I)\s$, implying that $g^r(x) \mid x^{\ell}-1$ with $e \cdot p^c = \ord \left( g^{r+1}(x) \right)$. This forces $\ell \geq e \cdot p^c$, which is a contradiction.\qed
\end{proof}

With $r$ not a power of $p$, let a state $\b$ of a sequence $\u$ with minimal polynomial $g^{r}(x)$ be given. Then, by (\ref{equ:defD}), the $q^n$ distinct vectors $\a $ satisfying $\b = {\rm D}(\a)$ are
states of pairwise distinct cycles, all of which has $g^{r+1}(x)$ as their minimal polynomial.

Given respective states $\b$ and $\b'$ from two distinct cycles, each with minimal polynomial $g^{r}(x)$,their ${\rm D}$ pre-image states $\a$ and $\a'$ belong to two distinct cycles with minimal polynomial $g^{r+1}(x)$. Thus, given all $\delta:=t \cdot q^{n(r-1)} \cdot p^{-c}$ states of distinct cycles with minimal polynomial $g^{r}(x)$, the exact $q^{n} \cdot \delta$ states of all distinct cycles having minimal polynomial $g^{r+1}(x)$ are known.

In fact, given $\b$, we can let $a_0=\ldots=a_{n-1}=0$ and construct a state $\a$ by using (\ref{equ:defD}). Taking the sum of $\a$ and each of the $n(r+1)$-stage states of sequences in $\Omega(g(x))$ gives us what we are after.

\begin{lemma}\label{lem:power}
Let $r=p^d$ for some nonnegative integer $d$. For a sequence $\s$ with minimal polynomial $g^{r+1}(x)$,
there exists a sequence $\s'$ in $\Omega(g(x))$ such that the set $\{\s+\eta\s':\eta\in\F_p\}$ contains $p$ shift equivalent sequences.
\end{lemma}
\begin{proof}
Lemma~\ref{lem:cycle} says that the period of $\s$ is $e \cdot p^{d+1}=p \cdot e \cdot r$.

Observe that that $L^{e \cdot r}\s-\s=(L^{e \cdot r}-I)\s=(L^e-I)^r \s$.
Since $g(x)$ is irreducible, $g^r(x)$ divides $(x^e-1)^r$ but $g^{r+1}(x)$ does not divide
$(x^e-1)^r$. Hence, $g(x)$ is the minimal polynomial of $\s':=(L^e-I)^r \s \neq \0$
and $L^{e \cdot r}\s = \s+ \s'$.

Let $2 \leq \eta \leq p-1$. Since $\s'$ has period $e$,
\[
L^{\eta \cdot e \cdot r} \s - L^{(\eta-1) \cdot e \cdot r}\s
=L^{(\eta-1) \cdot e \cdot r}[(L^{e \cdot r}-I)\s]=L^{(\eta -1) \cdot e \cdot r} \s'=\s'.
\]
Proceeding inductively, one obtains
$L^{\eta \cdot e \cdot r} \s=L^{(\eta-1) \cdot e \cdot r}\s+\s'=\s+ \eta \s'$.
\qed
\end{proof}

When $r=p^d$, given a state $\b$ of a sequence $\u$ with minimal polynomial $g^{r}(x)$,
we combine Lemma~\ref{lem:power} and (\ref{equ:defD}) to determine the $q^n$ distinct vectors $\a \in \F_q^{n(r+1)}$ satisfying $\b= {\rm D}(\a)$. For an arbitrary $\a$, let $\a':=T^{e \cdot r}\a-\a$.
Then the $p$ distinct states $\a+\eta\a' \text{ with } \eta \in \F_p$ correspond to one cycle with minimal polynomial $g^{r+1}(x)$.

For two arbitrary such states, say $\a_1$ and $\a_2$, it follows that
$\a'=T^{e \cdot r}\a_1-\a_1=T^{e \cdot r}\a_2-\a_2$. Thus, given $\b$ and $\a'$, we can construct the states of all the $\frac{q^n}{p}$ distinct cycles with minimal polynomial $g^{r+1}(x)$. Since there are $t \cdot \frac{q^{n(r-1)}}{r}$ possible choices for $\b$, we can provide the exact
$t \cdot q^{n \cdot r}  \cdot p^{-(d+1)}$ states of all the distinct cycles with minimal polynomial $g^{r+1}(x)$.

\begin{remark}
It may be possible to derive an exact formula for $\a'$ given $g(x)$, the value $r=p^d$, and a vector $\b$. The formula is likely to be quite complicated and we leave it as an open problem. Using the LFSR to compute $T^{e \cdot r}\a-\a$ directly from an arbitrary $\a$ seems to be the simplest approach.
\end{remark}

\begin{example}(\cite{Lempel70})\label{example2}
Let $g(x)=x+1 \in \F_2[x]$. Given a state $\b=(b_0,b_1,\ldots,b_{r-1})\in\F_2^r$ of a sequence with minimal polynomial $g^r(x)=(x+1)^r$ and an element $a_0 \in \F_2$, the entries of the resulting
$\a=(a_0,a_1,\ldots, a_r) \in \F_2^{r+1}$ are
$a_i = a_0 +b_0+ \ldots + b_{i-1}$ for $i \in \{1,2,\ldots,r\}$. When $a_0=0$, we denote the resulting state by $\a_0$. The state is $\a_1$ when $a_0=1$. Since the only nonzero sequence with minimal polynomial $x+1$ is the all one sequence $\1$, we know that $\a_1=\a_0+\1$.

If $r$ is not a power of $2$, then $\a_0$ and $\a_1$ are the respective states of two distinct sequences sharing the same minimal polynomial $(x+1)^{r+1}$. If $r$ is a power of $2$, then $\a_0$ and $\a_1$ belong to the same sequence with minimal polynomial $(x+1)^{r+1}$.\qed
\end{example}

\begin{example}\label{example3}
Let $g(x)=x^3+2x+2\in\F_3[x]$, and note that $r=1=3^0$. From Example~\ref{example1}, the
respective states of the two nonzero cycles are $(1,0,0)$ and $(2,1,1)$.

Using $\b=(1,0,0)$ and letting $a_0=a_1=a_2=0$, we get a $6$-stage
state $\a=(0,0,0,1,0,1)$ from (\ref{equ:defD}) and $\a'=T^{13} \a -\a=(1,2,0,0,2,0)$. 
The corresponding sequence $\s$ of period $39$ with minimal polynomial 
$g^2(x)=x^6 + x^4 + x^3 + x^2 + 2x+ 1$ is
\begin{equation}\label{eq:s}
(\underbrace{0,0,0,1,0,1}_{\a},2,1,1,1,0,1,0, \underbrace{1,2,0,1,2,1}_{\a+\a'},1,0,0,2,1,1,2, \underbrace{2,1,0,1,1,1}_{\a+2\a'},0,2,2,0,2,1,1).
\end{equation}
It is easy to check that $L^{13}\s-\s=L^{26}\s-L^{13}\s=(1,2,0,0,2,0,2,2,2,1,1,0,2)$.
As indicated by the underbraces in (\ref{eq:s}), $\a$, $\a+\a'$ and $\a+2\a'$
generate the same cycle $[\s]$.

The $27$ distinct $6$-stage states in $\ker({\rm D})$ can be partitioned into $9$ disjoint sets,
each of the form $\a_i+\{\0,\a',2\a'\}$ for $0\leq i<9$. The respective states of the cycles with minimal polynomial $g^2(x)$ are $\a+\a_i$. After some computation we arrive at
\begin{multline*}
\{\a_i:0\leq i <9\}=\{(\0), (1,0,0,1,0,1), (0,0,1,0,1,1), (1,0,1,1,1,2),(1,1,1,2,2,0),\\
(1,1,2,2,0,1), (1,2,2,0,1,2), (0,1,2,1,0,0), (1,1,0,2,1,2)\}.
\end{multline*}
Obtaining $\{\a+\a_i: 0\leq i <9\}$ means we are halfway done.

To derive the remaining half we use $\b=(2,1,1)$ and let $a_0 =
a_1 = a_2=0$. This yields a $6$-stage
state $\a=(0,0,0,2,1,0)$ with $\a'=(1,1,2,2,0,1)$. The corresponding
sequence $\s$ is
\[
(\underbrace{0,0,0,2,1,0}_{\a},0,0,1,2,1,0,2, \underbrace{1,1,2,1,1,1}_{\a+\a'},2,1,1,2,2,0,0, \underbrace{2,2,1,0,1,2}_{\a+2\a'},1,2,1,2,0,0,1).
\]
One eventually gets
\begin{multline*}
\{\a_i:0\leq i <9\}=\{(\0), (1,0,0,1,0,1), (0,0,1,0,1,1), (0,1,0,1,1,1),(1,0,1,1,1,2),\\
(0,1,1,1,2,2), (1,1,1,2,2,0), (0,1,2,1,0,0), (2,1,1,0,2,1)\}
\end{multline*}
from which we can explicitly compute elements in $\{\a+\a_i: 0\leq i <9\}$.

All of the $18$ states we have determined above belong to their $18$ respective distinct cycles 
in $\Omega\left(g^2(x)\right)\setminus \Omega(g(x))$.\qed
\end{example}

Applying the approach recursively, starting from the states of distinct cycles
in $\Omega(g(x))$, we determine the states of distinct cycles in $\Omega(g^b(x))$.
For brevity, we present them as $n \cdot b$-stage states. This step can be performed
by using well-known properties of the LFSRs or, better still, by the method that
we propose in the next section.

\section{The Cycle Structure of $\Omega(f(x))$}\label{sec:c3}
We combine the results established in previous sections to study $\Omega(f(x))$.
The first subsection considers the cycle structure of $\Omega(f(x))$. The second one develops
a procedure to find a state of each cycle in $\Omega(f(x))$.

\subsection{The Cycle Structure}
We start by setting up the notations. Recall that $f(x)=\prod_{i=1}^{k} g_i^{b_i}(x)$.
Let $g_i(x)$ be of degree $n_i$ having $\beta_i$ as a root.
Let $e_i = \ord(g_i(x))=\ord(\beta_i)$ and $t_i=\frac{q^{n_i}-1}{e_i}$.
Based on Lemma~\ref{lem:cycle}, we label the cycles in $\Omega(g_i^{b_i}(x))$:
\begin{equation}\label{eq:label}
[\0],[\u^i_0],[\u^i_1],\ldots,[\u^i_{\sigma_i-1}]
 \mbox{ with } 
\sigma_i=\sum_{r_i=1}^{b_i}t_i\frac{q^{n_i(r_i-1)}}{p^{c_{r_i}}}
\end{equation}
where $c_{r_i}$ is the least integer satisfying $p^{c_{r_i}}\geq r_i$.
The respective periods and states
of the cycles are $1,e^i_0,e^i_1,\ldots,e^i_{\sigma_i-1}$ and
$\0,\a^i_0,\a^i_1,\ldots,\a^i_{\sigma_i-1} \in \F_q^{b_i \cdot n_i}$.

The cycle structure of $\Omega(f(x))$ can be deduced from the three properties listed at the end of Section~\ref{sec:prelims}. We give the formal statement in the next theorem.

\begin{theorem}\label{thm:cycle-f}
Let $\gamma:=\gcd\left(e^{k}_{j_k},\lcm\left(e^1_{j_1},\ldots,e^{k-1}_{j_{k-1}}\right)\right)$.
Then
\begin{multline}\label{eq:longf}
\Omega(f(x))=[\0]\bigcup_{i=1}^{k} \bigcup_{j=0}^{\sigma_i-1}[\u^i_j]~
\bigcup_{1 \leq i_1 < i_2 \leq k}~\bigcup_{j_1=0}^{\sigma_{i_1}-1}~
\bigcup_{j_2=0}^{\sigma_{i_2}-1}~\bigcup_{\ell_2=0}^{\gcd \left(e^{i_2}_{j_2},e^{i_1}_{j_1}\right)-1}
\left[\u^{i_1}_{j_1}+L^{\ell_2}\u^{i_2}_{j_2}\right]
\cdots~\cdots \\
\bigcup_{j_1=0}^{\sigma_1-1} \cdots \bigcup_{j_k=0}^{\sigma_k-1}~
\bigcup_{\ell_2=0}^{\gcd\left(e^2_{j_2},e^1_{j_1}\right)-1}\cdots~
\bigcup_{\ell_{k}=0}^{\gamma-1}\left[\u_{j_1}^{1}+L^{\ell_2}\u_{j_2}^{2}+\ldots
+L^{\ell_{k}}\u_{j_k}^{k}\right].
\end{multline}
If some cycles being the same is allowed, we can express
$\Omega(f(x))$ more succinctly as
\begin{equation}\label{eq:shortf}
\Omega(f(x))=\bigcup_{\substack{a_i \in\{0,1\}\\ 1\leq  i \leq k}}~
\bigcup_{j_1=0}^{\sigma_1-1} \cdots ~\bigcup_{j_k=0}^{\sigma_k-1}~
\bigcup_{\ell_2=0}^{\gcd\left(e^2_{j_2},e^1_{j_1}\right)-1}~\cdots~
\bigcup_{\ell_{k}=0}^{\gamma-1}
\left[a_1 \u_{j_1}^{1}+a_2 L^{\ell_2} \u_{j_2}^{2}+\ldots+ a_k L^{\ell_{k}}\u_{j_k}^{k}\right].
\end{equation}
\end{theorem}

\begin{remark}
Theorem~\ref{thm:cycle-f} shows the {\it types} of nonzero cycles that we have in $\Omega(f(x))$. Cycles $[\u_j^i]$ come from
$\Omega\left(g_i^{b_i}(x)\right)$. For an $r$-subset $S_r:=\{i_1,i_2,\ldots,i_r\}$ of $\{1,2,\ldots,k\}$ with $i_j < i_{\ell}$ for $j < \ell$, cycles $\left[a_1\u_{j_1}^{i_1}+a_2L^{\ell_2}\u_{j_2}^{i_2}+\ldots+a_r L^{\ell_{r}}\u_{j_r}^{i_r}\right]$ are from $\Omega\left( \Pi_{i \in S_r} \left(g_i^{b_i}(x)\right)\right)$. If some of the
$a_1,a_2,\ldots,a_k$ are $0$, then
$\left[a_1\u_{j_1}^{1}+a_2L^{\ell_2}\u_{j_2}^{2}+\cdots+a_k L^{\ell_{k}}\u_{j_k}^{k}\right]$
may be the same for different parameter choices.
\end{remark}
\begin{proof}
When $k=1$, the statements hold.

Let $k=2$ and $f(x)=g_1^{b_1}(x)~g_2^{b_2}(x) \in \F_q[x]$. It is clear that $\Omega(f(x))$ contains $\Omega(g_1^{b_1}(x))$ and $\Omega(g_2^{b_2}(x))$ as subsets. Hence,
\[
[\0]~\bigcup_{j_1=0}^{\sigma_1-1}[\u^1_{j_1}]~\bigcup_{j_2=0}^{\sigma_2-1}[\u^2_{j_2}] \subseteq\Omega(f(x)).
\]
All other sequences in $\Omega(f(x))$ must be of the form
\begin{equation}\label{sum1}
L^{\ell_1}\u^1_{j_1}+L^{\ell_2}\u^2_{j_2}=L^{\ell_1}\left(\u^1_{j_1}+L^{\ell_2-\ell_1}\u^2_{j_2}\right)
\end{equation}
with $0 \leq j_1 < \delta_1$, $0 \leq j_2 <\delta_2$, $0 \leq \ell_1 <e^1_{j_1}$, $0\leq\ell_2<e^2_{j_2}$, and $\ell_2-\ell_1$ reduced modulo $e^2_{j_2}$. The period is $\lcm\left(e^1_{j_1},e^2_{j_2}\right)$.

When $j_1\neq j_1'$ or $j_2\neq j_2'$, sequences
$L^{\ell_1}\left(\u^1_{j_1}+L^{\ell_2-\ell_1}\u^2_{j_2}\right)$ and
$L^{\ell_1'}\left(\u^1_{j_1'}+L^{\ell_2'-\ell_1'}\u^2_{j_2'}\right)$ are
shift inequivalent.
For a contradiction, assume that for some $\ell_1,\ell_1',\ell_2$, and $\ell_2'$,
\[
L^{\ell_1}\u^1_{j_1}+L^{\ell_2}\u^2_{j_2} =L^{\ell_1'}\u^1_{j_1'}+L^{\ell_2'}\u^2_{j_2'} \iff\\
L^{\ell_1}\u^1_{j_1}-L^{\ell_1'}\u^1_{j_1'} =-(L^{\ell_2}\u^2_{j_2}-L^{\ell_2'}\u^2_{j_2'})=\0.
\]
The last equality holds since
\[
L^{\ell_1}\u^1_{j_1}-L^{\ell_1'}\u^1_{j_1'} \in \Omega\left(g_1^{b_1}(x)\right) \mbox{ and }L^{\ell_2}\u^2_{j_2}-L^{\ell_2'}\u^2_{j_2'} \in \Omega\left(g_2^{b_2}(x)\right),
\]
and $g_1(x)$ and $g_2(x)$ are distinct irreducible polynomials. This forces $j_1=j_1'$
and $j_2=j_2'$, a contradiction.

The $e^1_{j_1} \cdot e^2_{j_2}$ distinct sequences having the form of (\ref{sum1}) are divided into
$\frac{e^1_{j_1} \cdot e^2_{j_2}}{\lcm\left(e^1_{j_1},e^2_{j_2}\right)} =\gcd\left(e^1_{j_1},e^2_{j_2}\right)$ distinct cycles. Let 
$\tau_1:=\frac{e^1_{j_1}}{\gcd\left(e^1_{j_1},e^2_{j_2}\right)}$
and $\tau_2:=\frac{e^2_{j_2}}{\gcd\left(e^1_{j_1},e^2_{j_2}\right)}$.
We show that 
\[
\u^1_{j_1}+L^{\ell_2}\u^2_{j_2} \mbox{ and } \u^1_{j_1}+L^{\ell_2+ \kappa \cdot \gcd\left(e^1_{j_1},e^2_{j_2}\right)}\u^2_{j_2}
\]
are shift equivalent for $0 \leq \ell_2 < \gcd \left(e^1_{j_1},e^2_{j_2}\right)$ and $0 < \kappa <\tau_2$.

Since $\gcd(\tau_1,\tau_2)=1$, there exist $v,w \in \Z$ such that
$v \cdot \tau_1 + w \cdot \tau_2 = 1$, implying
\[
\kappa \cdot v \cdot e^1_{j_1}= \kappa \left(\gcd\left(e^1_{j_1},e^2_{j_2}\right)- w \cdot e^2_{j_2} \right).
\]
The respective periods of $\u^1_{j_1}$ and $\u^2_{j_2}$ are $e^1_{j_1}$ and $e^2_{j_2}$.
Hence,
\[
\u^1_{j_1}+L^{\ell_2+t \cdot \gcd\left(e^1_{j_1},e^2_{j_2}\right)}\u^2_{j_2}=
L^{\kappa \cdot v \cdot e^1_{j_1}} + L^{\ell_2+ \kappa \cdot \left(\gcd \left(e^1_{j_1},e^2_{j_2}\right)-w \cdot e^2_{j_2}\right)}\u^2_{j_2}=
L^{\kappa \cdot v \cdot e^1_{j_1}}(\u^1_{j_1}+L^{\ell_2}\u^2_{j_2}).
\]
Thus, a $(j_1, j_2)$ pair corresponds to the cycle
$\u^1_{j_1}+L^{\ell_2}\u^2_{j_2}$ with $0 \leq \ell_2 \leq \gcd\left(e^1_{j_1},e^2_{j_2}\right)-1$.

Going through all possible $(j_1, j_2)$ pairs gives us all the possible cycles, confirming
(\ref{eq:longf}) for the case of $k=2$. The case of $k>2$ follows by induction.\qed
\end{proof}

\subsection{A State Belonging to Each Cycle}
We now denote a cycle in $\Omega(f(x))$ by $[\c]$ with
$\c=a_1\u_{j_1}^{1}+a_2L^{\ell_2}\u_{j_2}^{2}+\ldots+a_k L^{\ell_{k}}\u_{j_k}^{k}$.
If $a_i=0$ for some $i$, then $a_i \u^i_{j_i}$ is the all zero sequence $\0$.

Let $n_i'=b_i \cdot n_i$ and $n=n_1'+n_2'+\cdots+n_s'$.
By Sections \ref{sec:c1} and \ref{sec:c2}, to find a state of $\c$,
it suffices to find an $n$-stage state $a_i\v^i_{j_i}\in\F_q^{n}$
of $a_i\u_{j_i}$ for all $i$. The desired state is
\begin{equation}\label{eq:represent}
a_1 \a_{j_1}^{1}+a_2 T^{\ell_2} \a_{j_2}^{2}+\ldots+ a_k T^{\ell_{k}}\a_{j_k}^{k}.
\end{equation}

The representation in (\ref{eq:represent}) has several limitations.
Given such a global expression, we can only use the properties of the
whole cycle $[\c]$ without being able to utilize the properties of
every {\it component sequence} $\u^i_{j_i}$. When the period of
$[\c]$ is fairly large, determining whether a state is in this
cycle may be hard to do, so we propose an alternative.

To each $g_i^{b_i}(x)$, one associates an $n_i'\times n_i'$ matrix $\A_i$ similar to the one constructed in (\ref{comx}), {\it i.e.}, treating $g_i^{b_i}(x)$ as a characteristic polynomial.
Construct the $n_i' \times n$ matrix $\P_i$ in the following manner. The first $n_i'$ columns form the identity matrix $I_{n_i'}$. The matrix $\A_i$ occupies columns $2$ to $n_i'+1$. Next, add columns recursively by appending the last column of $\A_{i}^{j}$, starting with $j=2$ until all $n$ columns of $\P_i$ are completed.

Given an $n_i'$-stage state $\a_i$ of some sequence in $\Omega(g_i^{b_i}(x))$, it is evident
that $\a_i \P_i$ is the corresponding $n$-stage state of that same sequence.
\begin{lemma}\label{lem:P}
The $n \times n$ matrix
\[
\P=
\begin{pmatrix}
\P_1\\
\P_2\\
\vdots\\
\P_k
\end{pmatrix}
\mbox {is of full rank, \ie, } \rank(\P)=n \mbox{, making } \P \in GL(n,\F_q).
\]
\end{lemma}
\begin{proof}
Let $\alb_{i,j}$ denote the $j$-th row of $\P_i$. We show that the rows $\alb_{i,j}$ of $\P$ are linearly independent for all $(i,j)$ with $1 \leq i \leq k$ and $1 \leq j \leq n_i'$.

Notice that $\alb_{i,j}$ is the first $n$ entries of the sequence from LFSR with characteristic polynomial
$g_i^{b_i}(x)$ and initial $n_i'$-stage state $(0,\ldots,0,1,0,\ldots,0)\in\F_q^{n_i'}$, where the unique $1$ is
in the $j$-th position. For a fixed $i$, it is clear that the rows of $\P_i$ are linearly independent.

For a contradiction, suppose that there exists a linear combination of the $n$ rows of $\P$
\[
\sum_{i=1}^{k} \sum_{j=1}^{n_i'} a_{i,j} \alb_{i,j}=\0
\]
with not all $a_{i,j}=0$. Without loss of generality, with not all $a_{1,j}=0$, write
\begin{equation}\label{eq:ind}
\sum_{j=1}^{n_1'} a_{1,j} \alb_{1,j} = \sum_{i=2}^{k} \sum_{j=1}^{n_i'} (-a_{i,j}) \alb_{i,j}.
\end{equation}
The left hand side of (\ref{eq:ind}) is the first $n$ entries of a sequence in $\Omega\left(g_1^{b_1}(x)\right)$ with a nonzero initial state of length $n_1'$ while the right hand side is the first $n$ entries of a sequence in $\Omega\left(g_2^{b_2}(x) \cdots g_k^{b_k}(x)\right)$ with a nonzero initial state of length $n-n_1'$.
Since $\deg\left(g_1^{b_1}(x)\right)$ and $\deg\left(g_2^{b_2}(x)\cdots g_k^{b_k}(x)\right)$ are $<n$, if the first $n$ entries of these two sequences are equal, then they must be the same sequence. Hence, there exists a sequence that simultaneously belongs to both $\Omega\left(g_1^{b_1}(x)\right)$ and $\Omega\left(g_2^{b_2}(x)\cdots g_k^{b_k}(x)\right)$. Since $g_1(x),\ldots,g_k(x)$ are pairwise distinct, this sequence must be $\0$, a contradiction.\qed
\end{proof}

Suppose that we already have $\P$. Let $\v \in \F_q^{n}$ and $\a_i \in \F_q^{n_i'}$
with $1 \leq i \leq k$ be the respective $n$-stage and $n_i'$-stage states of the
sequences in $\Omega(f(x))$ and $\Omega\left(g_i^{b_i}(x)\right)$. Then $\P$ provides a
one-to-one correspondence between $\v$ and $(\a_1,\a_2,\ldots,\a_k)$ via
$\v=(\a_1,\a_2,\ldots,\a_k) \P$. Since
\[
T \v =T[(\a_1,\a_2,\ldots,\a_k) \P]=(T\a_1,T\a_2,\ldots,T\a_k) \P,
\]
we know $\P$ and $T$ commute. Hence, any sequence
$\s \in \Omega(f(x))$ with an initial state $\v$ is the sum of sequences $\s_i$ from $\Omega\left(g_i^{b_i}(x)\right)$ with corresponding initial states $\a_i$ for
$1\leq i\leq k$. When convenient, we use $(\a_1,\a_2,\ldots,\a_k)$ to represent $\v$.

Assume that we have obtained the set $E_i:=\{\0,\a^i_0,\a^i_1,\ldots,\a^i_{\sigma_i-1}\}$,
where $\a^i_j$ is a state of cycle $[\s^i_j]$ and $\0$ the state of $[\0]$. This set contains
all of the states corresponding to the $\sigma_i +1$ distinct cycles in
$\Omega(g_i^{b_i}(x))$. Then
\begin{equation}\label{eq:state}
\v=(\a^1_{j_1},\a^2_{j_2},\ldots,\a^k_{j_k}) \P \mbox{ with } \a^i_{j_i}\in E_i
\end{equation}
can be taken as an initial state of a cycle $[\s] \in \Omega(f(x))$. Notice that
\[
\s=a_1\u^1_{j_1}+a_2\u^2_{j_2}+\cdots+a_k \u^k_{j_k}
\]
with $a_i=0$ if the corresponding component state is $\0$ and $a_i=1$ otherwise.

Let $\ell_i$ be a nonnegative integer for $1 \leq i \leq k$. By the state obtained
in (\ref{eq:state}) and by the properties of $\P$ and $T$, the state
$
\v'=\left(T^{\ell_1}\a^1_{j_1},T^{\ell_2}\a^2_{j_2},\ldots,T^{\ell_k}\a^k_{j_k}\right) \P
$
belongs to cycle
\[
\left[a_1 L^{\ell_1}\s^1_{j_1}+a_2 L^{\ell_2}\s^2_{j_2}+\cdots+ a_k L^{\ell_k}\s^k_{j_k}\right]=
\left[a_1\s^1_{j_1}+a_2L^{\ell_2-\ell_1}\s^2_{j_2}+\cdots+a_k L^{\ell_k-\ell_1}\s^k_{j_k}\right].
\]
This approach enables us to find a state belonging to any cycle in $\Omega(f(x))$.

\begin{example}\label{ex:ternary}
Consider the characteristic polynomial
\begin{align*}
f(x)&=\underbrace{(x^2+1)^2}_{=g_1^{2}(x)} \underbrace{(x^3+2x+2)^2}_{=g_2^{2}(x)} \underbrace{(x^3+x^2+2)^{2}}_{=g_3^{2}(x)}\in \F_3[x]\\
    &=(x^4+2x^2+1) (x^6 + x^4 + x^3 + x^2 + 2x+ 1) (x^6 + 2x^5 + x^4 + x^3 + x^2 + 1)\\
    &=x^{16} + 2 x^{15} + x^{14} + 2 x^{13} + x^{12} + x^{11} + x^{9} + x^8 + x^7 + x^5 + x^4 + 
    2x^3 + x^2 + 2x + 1.
\end{align*}
The parameters are $(e_1,t_1)=(4,2)$, $(e_2,t_2)=(e_3,t_3)=(13,2)$, and $b_i=2$ for all $i$.
Our $g_2(x)$ and $g_3(x)$ here are exactly the $g_2(x)$ and $g_3(x)$ in Example~\ref{example1}
so we can use its relevant results. We choose $p_1(x)=x^2+x+2$ as the associated primitive polynomial of $g_{1}(x)$.The two nonzero cycles in $\Omega(g_1(x))$ are $[(1,0,2,0)]$ and $[(1,1,2,2)]$.

With each respective $\A_{i}$ in the box, the matrix $\P$ is given by
\begin{align}\label{eq:P}
\P&=
\begin{pmatrix}
\P_1\\
\P_2\\
\P_3
\end{pmatrix} \mbox{ with } 
\P_1 =\begin{pmatrix}
    1 & \tikzmark{left}{0} & 0 & 0 & 2 & 0 & 2 & 0 & 0 & 0 & 1 & 0 & 1 & 0 & 0 & 0 \\
    0 & 1 & 0 & 0 & 0 & 2 & 0 & 2 & 0 & 0 & 0 & 1 & 0 & 1 & 0 & 0 \\
    0 & 0 & 1 & 0 & 1 & 0 & 0 & 0 & 2 & 0 & 2 & 0 & 0 & 0 & 1 & 0 \\
    0 & 0 & 0 & 1 & \tikzmark{right}{0} & 1 & 0 & 0 & 0 & 2 & 0 & 2 & 0 & 0 & 0 & 1
\Highlight[first]
\end{pmatrix} \mbox{,} \notag\\
\P_2 &=\begin{pmatrix}
    1 & \tikzmark{left}{0} & 0 & 0 & 0 & 0 & 2 & 0 & 1 & 1 & 0 & 0 & 2 & 0 & 1 & 0 \\
    0 & 1 & 0 & 0 & 0 & 0 & 1 & 2 & 2 & 0 & 1 & 0 & 1 & 2 & 2 & 1 \\
    0 & 0 & 1 & 0 & 0 & 0 & 2 & 1 & 0 & 0 & 0 & 1 & 2 & 1 & 0 & 2 \\
    0 & 0 & 0 & 1 & 0 & 0 & 2 & 2 & 2 & 1 & 0 & 0 & 0 & 2 & 2 & 0 \\
    0 & 0 & 0 & 0 & 1 & 0 & 2 & 2 & 0 & 0 & 1 & 0 & 2 & 0 & 0 & 2 \\
    0 & 0 & 0 & 0 & 0 & 1 & \tikzmark{right}{0} & 2 & 2 & 0 & 0 & 1 & 0 & 2 & 0 & 0
\Highlight[second]
    \end{pmatrix} \mbox{, }
\P_3 =\begin{pmatrix}
    1 & \tikzmark{left}{0} & 0 & 0 & 0 & 0 & 2 & 2 & 0 & 2 & 1 & 0 & 1 & 2 & 0 & 1 \\
    0 & 1 & 0 & 0 & 0 & 0 & 0 & 2 & 2 & 0 & 2 & 1 & 0 & 1 & 2 & 0 \\
    0 & 0 & 1 & 0 & 0 & 0 & 2 & 2 & 2 & 1 & 1 & 2 & 2 & 2 & 1 & 0 \\
    0 & 0 & 0 & 1 & 0 & 0 & 2 & 1 & 2 & 1 & 2 & 1 & 0 & 1 & 2 & 2 \\
    0 & 0 & 0 & 0 & 1 & 0 & 2 & 1 & 1 & 1 & 2 & 2 & 2 & 2 & 1 & 0 \\
    0 & 0 & 0 & 0 & 0 & 1 & \tikzmark{right}{1} & 0 & 1 & 2 & 0 & 2 & 1 & 0 & 2 & 0
\Highlight[third]
\end{pmatrix}.
\end{align}

\begin{table}[h!]
\vspace{-0.2cm}
\caption{States and Nonzero Cycles in $\Omega(g_i^2(x))$}
\label{table:cycles}
\tiny
\centering
\renewcommand{\arraystretch}{1.5}
\begin{tabular}{c|ll}
\hline
No. & State & Nonzero Cycle \\
\hline
$1$ & $\a_0^1=(1,0)$ & $\left[\u_0^1=(1,0,2,0)\right]$\\

$2$ &  $\a_1^1=(1,1)$ & $\left[\u_1^1=(1,1,2,2)\right]$\\

$3$ &  $\a_2^1=(1,0,0,0)$ & $\left[\u_2^1=(1,0,0,0,2,0,2,0,0,0,1,0)\right]$\\

$4$ &   $\a_3^1=(1,2,0,1)$ & $\left[\u_3^1=(1,2,0,1,2,2,2,1, 0,2,1,1)\right]$\\

$5$ &  $\a_4^1=(1,1,0,2)$ & $\left[\u_4^1=(1,1,0,2, 2,1,2,2, 0,1,1,2)\right]$\\

$6$ &   $\a_5^1=(1,1,0,0)$ & $\left[\u_5^1=(1,1,0,0, 2,2,2,2, 0,0,1,1)\right]$\\

$7$ &  $\a_6^1=(2,1,2,0)$ & $\left[\u_6^1=(2,1,2,0, 0,2,1,2, 1,0,0,1)\right]$\\

$8$ &   $\a_7^1=(1,2,0,2)$ & $\left[\u_7^1=(1,2,0,2, 2,0,2,1, 0,1,1,0)\right]$\\
\hline

$1$ & $\a_0^2=(1,1,1)$ & $\left[\u_0^2=(1, 1, 1, 2, 2, 0, 1, 2, 1, 0, 0, 1, 0)\right]$\\

$2$ &  $\a_1^2=(1,2,0)$ & $\left[\u_1^2=(1, 2, 0, 0, 2, 0, 2, 2, 2, 1, 1, 0, 2)\right]$\\

$3$ &  $\a_2^2=(0,0,0,1,0,1)$ & $\left[\u_2^2=(0,0,0,1,0,1,2,1,1,1,0,1,0,1,2,0,1,2,1,1,0,0,2,1,1,2,2,1,0,1,1,1,0,2,2,0,2,1,1)\right]$\\

$4$ &   $\a_3^2=(1,0,0,2,0,2)$ & $\left[\u_3^2=(1,0,0,2,0,2,0,2,0,0,0,2,2,2,2,0,2,2,2,2,1,2,1,1,2,1,0,1,0,2,1,2,1,0,1,2,2,2,0)\right]$\\

$5$ &  $\a_4^2=(0,0,1,1,1,2)$ & $\left[\u_4^2=(0,0,1,1,1,2,0,0,0,1,1,0,1,1,2,1,1,0,2,2,2,2,2,2,0,0,2,1,1,1,2,2,1,1,1,0,0,0,2)\right]$\\

$6$ &   $\a_5^2=(1,0,1,2,1,0)$ & $\left[\u_5^2=(1,0,1,2,1,0,1,1,2,0,1,1,0,2,2,1,2,0,0,0,0,1,1,2,1,2,0,1,1,2,2,0,2,2,0,2,0,1,1)\right]$\\

$7$ &  $\a_6^2=(1,1,1,0,2,1)$ & $\left[\u_6^2=(1,1,1,0,2,1,0,0,2,1,0,2,0,2,0,1,0,1,1,2,2,1,2,1,2,2,0,2,1,0,0,1,1,1,0,0,2,2,1)\right]$\\

$8$ &   $\a_7^2=(1,1,2,0,0,2)$ & $\left[\u_7^2=(1,1,2,0,0,2,1,2,1,1,1,1,1,2,0,2,0,2,2,0,1,0,2,2,1,0,0,2,2,0,1,2,2,0,2,0,0,1,2)\right]$\\

$9$ &  $\a_8^2=(1,2,2,1,1,0)$ & $\left[\u_8^2=(1,2,2,1,1,0,0,1,1,2,0,2,1,2,1,2,1,0,0,2,0,0,0,1,2,0,0,0,2,1,2,0,1,2,2,1,2,2,2)\right]$\\

$10$ &  $\a_9^2=(0,1,2,2,0,1)$ & $\left[\u_9^2=(0,1,2,2,0,1,0,1,2,2,1,0,2,1,0,2,2,2,1,2,0,1,0,2,0,1,2,2,2,2,1,1,1,2,0,1,0,0,0)\right]$\\

$11$ & $\a_{10}^2=(1,1,0,0,1,0)$ & $\left[\u_{10}^2=(1,1,0,0,1,0,2,1,0,1,2,0,2,2,0,0,0,0,0,1,0,2,2,0,0,1,0,2,0,0,2,0,0,2,1,0,1,0,0)\right]$\\

$12$ & $\a_{11}^2=(0,0,0,2,1,0)$ & $\left[\u_{11}^2=(0,0,0,2,1,0,0,0,1,2,1,0,2,1,1,2,1,1,1,2,1,1,2,2,0,0,2,2,1,0,1,2,1,2,1,2,0,0,1)\right]$\\

$13$&$\a_{12}^2=(1,0,0,0,1,1)$ & $\left[\u_{12}^2=(1,0,0,0,1,1,1,1,0,1,1,1,1,2,1,2,2,1,2,0,2,0,1,2,1,2,0,2,1,1,1,0,2,0,0,1,0,1,0)\right]$\\

$14$ &$\a_{13}^2=(0,0,1,2,2,1)$ & $\left[\u_{13}^2=(0,0,1,2,2,1,1,2,0,2,2,2,0,1,1,0,1,2,2,0,0,0,2,0,2,1,2,2,2,0,2,0,2,1,0,2,1,2,2)\right]$\\

$15$ &$\a_{14}^2=(0,1,0,0,2,1)$ & $\left[\u_{14}^2=(0,1,0,0,2,1,2,2,1,0,0,1,2,1,2,2,2,2,2,1,0,1,0,1,1,0,2,0,1,1,2,0,0,1,1,0,2,1,1)\right]$\\

$16$ &$\a_{15}^2=(1,0,1,0,2,2)$ & $\left[\u_{15}^2=(1,0,1,0,2,2,2,0,2,1,2,0,2,2,1,0,2,2,0,1,1,2,1,0,0,0,0,2,2,1,2,1,0,2,2,1,1,0,1)\right]$\\

$17$ & $\a_{16}^2=(0,1,1,0,0,2)$ & $\left[\u_{16}^2=(0,1,1,0,0,2,0,1,0,0,1,0,0,1,2,0,2,0,0,2,2,0,0,2,0,1,2,0,2,1,0,1,1,0,0,0,0,0,2)\right]$\\

$18$ & $\a_{17}^2=(1,1,1,1,0,0)$ & $\left[\u_{17}^2=(1,1,1,1,0,0,1,2,2,2,1,1,2,2,2,0,0,0,1,0,0,2,2,2,1,0,0,0,2,2,0,2,2,1,2,2,0,1,1)\right]$\\

$19$ & $\a_{18}^2=(0,1,2,0,1,0)$ & $\left[\u_{18}^2=(0,1,2,0,1,0,1,0,2,0,2,2,1,1,2,1,2,1,1,0,1,2,0,0,2,2,2,0,0,1,1,2,2,2,2,0,1,2,0)\right]$\\

$20$ & $\a_{19}^2=(2,1,1,2,0,1)$ & $\left[\u_{19}^2=(2,1,1,2,0,1,2,0,1,1,1,2,1,0,2,0,1,0,2,1,1,1,1,2,2,2,1,0,2,0,0,0,0,2,1,1,0,2,0)\right]$\\

\hline
$1$ & $\a_0^3=(1,1,0)$ & $\left[\u_0^3=(1, 1, 0, 1, 0, 0, 1, 2, 1, 0, 2, 2, 1)\right]$\\

$2$ &  $\a_1^3=(1,2,0)$ & $\left[\u_1^3=(1, 2, 0, 1, 1, 2, 2, 2, 0, 2, 0, 0, 2)\right]$\\

$3$ &  $\a_2^3=(0,0,0,1,0,0)$ & $\left[\u_2^3=(0,0,0,1,0,0,2,1,2,1,2,1,0,1,2,2,0,0,2,2,1,1,2,1,1,1,2,1,1,2,0,1,2,1,0,0,0,1,2)\right]$\\

$4$ &   $\a_3^3=(1,1,0,2,0,0)$ & $\left[\u_3^3=(1,1,0,2,0,0,0,0,0,1,1,0,1,2,0,2,1,0,2,0,0,2,2,0,0,2,0,2,1,0,0,1,0,0,1,0,2,0,0)\right]$\\

$5$ &  $\a_4^3=(1,0,1,1,0,1)$ & $\left[\u_4^3=(1,0,1,1,0,1,1,2,2,0,1,2,1,2,2,0,0,0,0,1,2,1,1,0,2,2,0,1,2,2,0,2,1,2,0,2,2,2,0)\right]$\\

$6$ &   $\a_5^3=(0,1,0,1,1,2)$ & $\left[\u_5^3=(0,1,0,1,1,2,0,1,1,0,0,2,1,1,0,2,0,1,1,0,1,0,1,2,2,2,2,2,1,2,1,0,0,1,2,2,1,2,0)\right]$\\

$7$ &  $\a_6^3=(1,0,0,2,2,1)$ & $\left[\u_6^3=(1,0,0,2,2,1,2,0,1,2,0,2,0,2,2,2,1,2,0,2,0,0,0,2,2,1,0,1,1,0,2,2,2,0,2,1,1,2,2)\right]$\\

$8$ &   $\a_7^3=(0,0,1,0,1,0)$ & $\left[\u_7^3=(0,0,1,0,1,0,1,0,0,2,0,1,1,1,2,0,2,1,2,1,0,2,0,2,1,2,2,1,2,1,1,1,1,0,1,1,1,1,0)\right]$\\

$9$ &  $\a_8^3=(1,2,1,1,2,2)$ & $\left[\u_8^3=(1,2,1,1,2,2,0,2,0,1,0,1,0,2,1,0,0,2,1,0,2,2,2,2,1,1,0,0,2,2,2,0,0,2,1,0,1,1,2)\right]$\\

$10$ &  $\a_9^3=(0,2,2,2,1,1)$ & $\left[\u_9^3=(0,2,2,2,1,1,2,2,2,1,0,0,1,1,1,1,1,1,0,2,2,1,2,2,0,2,2,0,0,0,1,2,2,2,0,0,1,0,0)\right]$\\

$11$ & $\a_{10}^3=(2,0,1,2,2,2)$ & $\left[\u_{10}^3=(2,0,1,2,2,2,1,1,1,1,2,0,1,0,2,0,1,2,1,1,1,0,2,1,0,2,1,1,2,0,2,0,1,1,2,0,0,0,0)\right]$\\

$12$ & $\a_{11}^3=(0,0,0,1,1,2)$ & $\left[\u_{11}^3=(0,0,0,1,1,2,0,2,2,0,1,1,1,0,1,2,2,1,1,2,0,0,1,1,2,1,0,2,1,0,1,0,1,1,1,2,1,0,1)\right]$\\

$13$&$\a_{12}^3=(1,1,0,2,1,2)$ & $\left[\u_{12}^3=(1,1,0,2,1,2,1,1,0,0,0,0,2,1,2,2,0,1,1,0,2,1,1,0,1,2,1,0,1,1,1,0,2,0,2,2,0,2,2)\right]$\\

$14$ &$\a_{13}^3=(0,1,0,1,2,1)$ & $\left[\u_{13}^3=(0,1,0,1,2,1,1,2,1,2,2,2,2,0,2,2,2,2,0,0,0,2,0,2,0,2,0,0,1,0,2,2,2,1,0,1,2,1,2)\right]$\\

$15$ &$\a_{14}^3=(1,0,0,2,0,0)$ & $\left[\u_{14}^3=(1,0,0,2,0,0,0,1,1,1,2,2,1,1,1,2,0,0,2,2,2,2,2,2,0,1,1,2,1,1,0,1,1,0,0,0,2,1,1)\right]$\\

$16$ &$\a_{15}^3=(2,1,0,0,0,0)$ & $\left[\u_{15}^3=(2,1,0,0,0,0,1,0,2,1,1,1,2,2,2,2,1,0,2,0,1,0,2,1,2,2,2,0,1,2,0,1,2,2,1,0,1,0,2)\right]$\\

$17$ & $\a_{16}^3=(1,0,2,0,2,0)$ & $\left[\u_{16}^3=(1,0,2,0,2,0,1,2,0,0,1,2,0,1,1,1,1,2,2,0,0,1,1,1,0,0,1,2,0,2,2,1,2,1,2,2,1,1,0)\right]$\\

$18$ & $\a_{17}^3=(2,2,1,2,2,2)$ & $\left[\u_{17}^3=(2,2,1,2,2,2,1,2,2,1,0,2,1,2,0,0,0,2,1,0,0,0,2,0,0,1,2,1,2,1,2,0,2,1,1,0,0,1,1)
\right]$\\

$19$ & $\a_{18}^3=(2,1,1,2,1,0)$ & $\left[\u_{18}^3=(2,1,1,2,1,0,0,2,0,2,2,1,0,2,2,0,0,1,2,2,0,1,0,2,2,0,2,0,2,1,1,1,1,1,2,1,2,0,0)\right]$\\

$20$ & $\a_{19}^3=(0,1,1,0,0,1)$ & $\left[\u_{19}^3=(0,1,1,0,0,1,0,1,2,0,0,2,0,0,2,0,1,0,0,2,2,0,1,0,0,0,0,0,2,2,0,2,1,0,1,2,0,1,0)\right]$\\
\hline
\end{tabular}
\end{table}

There are $8$ nonzero cycles in $\Omega(g_1^2(x))$: $2$ cycles of period $4$ in $\Omega(g_1(x))$ 
and $6$ cycles, each with period $12$, in $\Omega(g_1^2(x)) \setminus \Omega(g_1(x))$. 
Performing the required computations, the first $3$ cycles with period $12$ can be derived using $\b=(1,0)$, 
$\a=(1,0,0,0)$, and $\a'=T^4 \a -\a =(1,0,2,0)$, yielding 
$\a_i \in \{(\0),(0,2,0,1), (0,1,0,2)\}$. The $4$-stage states that we want are $\{ \a + \a_i\}$. 
To get the other $3$, we use $\b=(1,1)$, $\a=(1,1,0,0)$, and $\a'=T^4 \a -\a =(1,1,2,2)$, yielding 
$\a_i \in \{(\0),(1,0,2,0), (0,1,0,2)\}$. Explicitly determining elements in 
$\{ \a + \a_i\}$ completes our task for $E_1$.

Examples~\ref{example1} and~\ref{example3} lead us to the nonzero cycles in $\Omega(g_2^2(x))$. 
There are $2$ cycles of period $13$ in $\Omega(g_2(x))$ and $18$ cycles, each with period $39$, in 
$\Omega(g_2^2(x)) \setminus \Omega(g_2(x))$.

To determine all nonzero cycles in $\Omega(g_3^2(x))$ we perform a similar analysis. Table~\ref{table:ex} already listed the $2$ cycles of period $13$ in $\Omega(g_3(x))$. The remaining $18$ cycles, each with period $39$, can be easily determined. The first $9$ comes from $\b=(1,1,0)$, $\a=(0,0,0,1,0,0)$, and $\a'=T^{13} \a - \a =(1,2,2,2,0,2)$. The initial states are $\a+\a_i$ with 
\begin{multline*}
\a_i \in \{(\0),(1,1,0,1,0,0),(1,0,1,0,0,1),(0,1,0,0,1,2),(1,0,0,1,2,1),\\
(0,0,1,2,1,0),(1,2,1,0,2,2),(0,2,2,1,1,1),(2,0,1,1,2,2)\}.
\end{multline*}
The remaining $9$ comes from $\b=(1,2,0)$, $\a=(0,0,0,1,1,2)$, and 
$\a'=(0,1,2,1,0,2)$. The initial states are $\a+\a_i$ with 
\begin{multline*}
\a_i \in \{(\0),(1,1,0,1,0,0),(0,1,0,0,1,2),(1,0,0,1,2,1),(2,1,0,2,2,1),\\
(1,0,2,2,1,1),(2,2,1,1,1,0),(2,1,1,1,0,1),(0,1,1,2,2,2)\}.
\end{multline*}

The complete list of nonzero cycles and their corresponding states is in Table~\ref{table:cycles}.

Consulting (\ref{eq:state}), we know that
\begin{align*}
\v = (\a_3^1,\a_9^2,\a_8^3) \P &= (1,2,0,1,0,1,2,2,0,1,1,2,1,1,2,2)\P \\
   &=(2,2,0,1,1,2,2,1,2,2,2,2,0,2,1,0)
\end{align*}
is a state of $[\s=\u_3^1+\u_9^2+\u_8^3] \in \Omega(f(x))$. 
The period of $\s$ is $\lcm(12,39,39)=156$. One can quickly verify that using $\v$ 
as the input state to the LFSR with characteristic polynomial $f(x)$ indeed yields $\s$ 
since the first $32$ entries of the two sequences match. \qed
\end{example}

We now highlight the advantage of our new representation of the states of the cycles in 
$\Omega(f(x))$. To construct de Bruijn sequences by the cycle joining 
method \cite{Fred82} one must find the conjugate pairs between any two cycles $C_1$ and 
$C_2$ in $\Omega(f(x))$. We transform this problem into one that decides on whether two 
states belong to the same cycle. The latter can be solved by applying the state shift operator 
$T$ repeatedly until one state is shown to be the other's cyclic shift or not.

The naive approach is to do exhaustive searching. If the period of the cycles are large, 
this is time consuming. Using the new representation, 
Algorithm~\ref{algo:state} simplifies the work.

\begin{algorithm}[h!]
\caption{Determining if two states belong to the same cycle}
\label{algo:state}
\begin{algorithmic}[1]
 \renewcommand{\algorithmicrequire}{\textbf{Input:}}
 \renewcommand{\algorithmicensure}{\textbf{Output:}}
 \Require $\P$ and two states $\v_1 \neq \v_2$.
 \Ensure Whether $\v_1, \v_2$ are in the same cycle.

\Procedure {Main~}{$\P$, $\v_1$, $\v_2$}
\State $(\a_1,\a_2,\ldots,\a_k) \gets \v_1\P^{-1}$
\State $(\b_1,\b_2,\ldots,\b_k) \gets \v_2\P^{-1}$

\For{$i$ from $1$ to $k$} \Comment{Finding the period of the cycle containing state $\a_i$}
  \State{$\a \gets \a_i$, $e_i \gets 0$}
  \While{$T\a \neq \a_i$}
    \State{$e_i \gets e_i+1$}
  \EndWhile
\EndFor 

\For{$i$ from $1$ to $k$} \Comment{Find whether $T^{\ell_i}\a_i=\b_i$}
  \State{$\a \gets \a_i$, $\ell_i \gets 0$}
  \If{$\a \neq \b_i$}
    \State{$\a \gets T \a$, $\ell_i \gets 1$}
  \EndIf  
  \While{$\a \neq \b_i$ and $\a \neq \a_i$}
    \State{$\a \gets T\a$, $\ell_i \gets \ell_i+1$}
  \EndWhile
  \If{$\a \neq \b_i$}
    \State{Output NO and break}
  \EndIf
\EndFor 

\For{$i$ from 2 to $k$} \Comment{by Generalized CRT}
  \For{$j$ from $1$ to $i-1$}
    \If{$\gcd(e_i,e_j) \mbox{ does not divide } \ell_i-\ell_j$}
      \State{Output NO and break}
    \EndIf
  \EndFor
\EndFor  
\State{Output: YES}

\EndProcedure
\end{algorithmic}
\end{algorithm}

\begin{theorem}\label{th:alg}
Algorithm~\ref{algo:state} is correct.
\end{theorem}

\begin{proof}
If the input states $\v_1$ and $\v_2$ in $\Omega(f(x))$ are in the same cycle, then 
there exists an integer $\ell$ such that $\v_1= T^{\ell} \v_2$. Multiplying by $\P^{-1}$, 
we have $\v_1 \P^{-1}=(\a_1,\a_2,\ldots,\a_k)$ and $\v_2\P^{-1}=(\b_1,\b_2,\ldots,\b_k)$ with 
$\a_i$ and $\b_i$ in $\Omega(g_i^{b_i}(x))$. If $\v_1$ and $\v_2$ are in the same cycle,
then $\a_i$ and $\b_i$ must be in the same cycle in $\Omega(g_i^{b_i}(x))$ for all $1 \leq i \leq k$.

The first part of Algorithm~\ref{algo:state} determines the period $e_i$ of the cycle containing the state $\a_i$. The second part tests whether $\a_i$ and $\b_i$ are in the same ``component cycle'' in 
$\Omega(g_i^{b_i}(x))$. If the test fail for at least one $i$ then $\v_1$ and $\v_2$ belong to 
distinct cycles. Otherwise, we obtain integers $\ell_i$s that satisfy $T^{\ell_i} \a_i= \b_i$ for each $i$. The last part verifies if there exists an integer $\ell$ that makes $\v_1=T^{\ell} \v_2$. If such an $\ell$ exists, then $\v_1\P^{-1}=T^{\ell}\v_2\P^{-1}$, \ie,
\[
(T^{\ell_1}\a_1,T^{\ell_2}\a_2,\ldots,T^{\ell_k}\a_k)=(T^{\ell}\b_1,T^{\ell}\b_2,\ldots,T^{\ell}\b_k).
\]
Hence, $\ell$ exists if and only if there exists a solution to the system of congruences:
\begin{equation}\label{eq:cong}
\begin{cases}
\ell \equiv \ell_1 \Mod{e_1} \\
\ell \equiv \ell_2 \Mod{e_2} \\
\cdots\\
\ell \equiv \ell_k \Mod{e_k}
\end{cases}.
\end{equation}
We know from Generalized Chinese Remainder Theorem (CRT)~\cite[Thm.~2.4.2]{Ding96} that 
(\ref{eq:cong}) has a solution if and only if the following equations hold simultaneously:
\[
\gcd(e_i,e_j) \mbox{ divides } \ell_i-\ell_j \mbox{ for all } 1\leq i\neq j\leq k.
\]
\qed
\end{proof}

Our new representation allows for a faster check on whether a state $\v$ is in some $[\s]$ by 
inputting $\v$ and any state $\v'$ belonging to $[\s]$ as $\v_1$ and $\v_2$ in Algorithm~\ref{algo:state}. Recall that the number of states in $[\s]$ is the period of $\s$. Instead of comparing $\v$ with all possible states in $[\s]$, the representation transforms the problem into performing the verification in $k$ ``component cycles''. Algorithm~\ref{algo:state} requires at most 
$\sum_{i=1}^k e_i$ steps while an exhaustive search takes $\lcm\{e_1,\ldots,e_k\}$ steps to complete. 

\begin{example}
The input consists of $\P$ in (\ref{eq:P}), $\v_1=(2,2,0,1,1,2,2,1,2,2,2,2,0,2,1,0)$, and 
$\v_2=(0,1,0,1,0,2,0,1,0,1,1,0,0,1,1,2)$. Hence, $\a_1=(1,2,0,1)$, $\a_2=(0,1,2,2,0,1)$, and 
$\a_3=(1,2,1,1,2,2)$ while $\b_1=(0,1,0,2)$, $\b_2=(1,1,0,2,0,0)$, and $\b_3=(2,2,0,0,0,1)$. 
The sequence generated by the LFSR with characteristic polynomial $g_1^2(x)$ on input $\a_1$ never 
contains $\b_1$ as a state. Hence, $\v_1$ and $\v_2$ are never in a common cycle.

Keep the same $\P$ and $\v_1$ but with $\v_2=(1,0,2,0,2,2,0,1,1,2,2,1,2,2,2,2)$. Hence, $\b_1=(0,2,1,1)$, $\b_2=(1,0,0,0,0,1)$, and $\b_3=(0,1,1,2,1,2)$. The sequence generated by the LFSR with characteristic polynomial $g_1^2(x)$ on input $\a_1$ is 
$(1,2,0,1,2,2,2,1,0,2,1,1)$, making $T^8 \a_1=\b_1$. Similarly, one obtains $T^{35} \a_2=\b_2$ and 
$T^{35} \a_3=\b_3$. Thus, $\ell_1=8$ and $\ell_2=\ell_3=35$. Since  $e_1=12$ and $e_2=e_3=39$, it is immediate to confirm that $\gcd(39,12)=3$ divides $\ell_3-\ell_1=\ell_2-\ell_1=27$ and $\gcd(e_3,e_2)=39$ divides $\ell_3-\ell_2=0$, ensuring $\v_1$ and $\v_2$ belong to the same cycle.
\qed
\end{example}

\section{Conclusions}\label{sec:conclu}
Theorem~\ref{thm:cycle-f} presents the cycle structure of LFSRs with arbitrary characteristic polynomial $f(x)\in\F_q[x]$. This had not been previously done in the literature. We put forward a method to find a state of each
cycle in $\Omega(f(x))$ by devising a new representation of the states to expedite the verification process. Finding conjugate pairs shared by two arbitrary cycles, which is crucial in the cycle joining method, can then be done more efficiently.

\begin{acknowledgements}
The work of Z.~Chang is supported by the Joint Fund of the National Natural Science Foundation of China under Grant U1304604. Research Grants TL-9014101684-01 and MOE2013-T2-1-041 support the research carried out by M.~F.~Ezerman, S.~Ling, and H.~Wang.
\end{acknowledgements}

\bibliographystyle{plain}

\end{document}